\documentclass[english,smallextended]{article}
\usepackage[T1]{fontenc}
\usepackage[latin9]{inputenc}
\usepackage{babel}
\usepackage{array}
\usepackage{amsmath}
\usepackage{amssymb}
\usepackage{stmaryrd}
\usepackage{graphicx}
\usepackage[unicode=true,
 bookmarks=false,
 breaklinks=false,pdfborder={0 0 1},backref=false,colorlinks=false]
 {hyperref}
\hypersetup{pdftitle={\@shorttitle},
 pdfauthor={\@shortauthor},
 pdfsubject={\@Subject},
 pdfkeywords={Computer Graphics Forum, EUROGRAPHICS},
   colorlinks,linkcolor=blue,citecolor=blue,urlcolor=blue,    bookmarks=false,    pdfpagemode=UseNone}

\makeatletter

\providecommand{\tabularnewline}{\\}

\usepackage{arxiv}

\usepackage{hyperref}       
\usepackage{url}            
\usepackage{booktabs}       
\usepackage{amsfonts}       
\usepackage{nicefrac}       
\usepackage{microtype}      

\usepackage{amsthm}
\newtheorem{claim}{Claim}
\usepackage{tikz,tkz-tab,tkz-graph}
\usepackage{relsize}

\usetikzlibrary{arrows,automata}
\usetikzlibrary{positioning, fadings,snakes,calc,fadings,shapes.arrows,shadows,backgrounds}
\usetikzlibrary{shapes.geometric}
\tikzset{
    state/.style={
           rectangle,
  fill=#1!5!white,
           draw=#1, very thick,
           minimum height=2em,
           inner sep=2pt,
           text centered,
           },
    highlight/.style={
           rectangle,
  fill=#1!50!white,
           rounded corners,
           draw=#1, very thick,
           minimum height=2em,
           inner sep=2pt,
           text centered,
           },
coeff/.style={
           circle,
           draw=black, very thick,
           minimum height=2em,
           inner sep=2pt,
           text centered,
           },
ptNode/.style={circle, fill=black,thick, inner sep=2pt, minimum size=0.2cm}	,
square/.style={regular polygon,regular polygon sides=4},
ptNodeSq/.style={square, fill=black,thick, inner sep=2pt, minimum size=0.2cm},	
}

\tikzset{
    invisible/.style={opacity=0},
    visible on/.style={alt=#1{}{invisible}},
    alt/.code args={<#1>#2#3}{%
      \alt<#1>{\pgfkeysalso{#2}}{\pgfkeysalso{#3}} 
    },
  }

\tikzset{fontscale/.style = {font=\relsize{#1}}
    }

\usepackage{colortbl}
\definecolor{lightgray}{gray}{0.9}

\definecolor{S_purple}{RGB}{204, 0, 204}
\definecolor{S_brique}{RGB}{204, 51, 0}
\definecolor{S_petrol}{RGB}{0, 102, 153}
\definecolor{S_green}{RGB}{0, 153, 0}

\definecolor{Special_blue}{rgb}{0.03, 0.35, 0.49}

\makeatother

\begin{document}
\title{Tuning Ranking in Co-occurrence Networks with General Biased Exchange-based Diffusion on Hyper-bag-graphs}
\author{
	Xavier Ouvrard\thanks{xavier.ouvrard@cern.ch}\\
CERN,\\ 
Esplanade des Particules, 1,\\ 
CH-1211 Meyrin (Switzerland)\\
	\And 
Jean-Marie Le~Goff\\
CERN,\\ Esplanade des Particules, 1,\\ CH-1211 Meyrin (Switzerland)\\
	\And 
St{\'e}phane Marchand-Maillet\\
University of Geneva,\\CUI, Battelle (Bat A),\\ Route de Drize, 7,\\ CH-1227 Carouge (Switzerland)
}
\maketitle
\begin{abstract}
Co-occurence networks can be adequately modeled by hyper-bag-graphs (hb-graphs for short). A hb-graph is a family of multisets having same universe, called the vertex set. An efficient exchange-based diffusion scheme has been previously proposed that allows the ranking of both vertices and hb-edges. In this article, we extend this scheme to allow biases of different kinds and explore their effect on the different rankings obtained. The biases enhance the emphasize on some particular aspects of the network.
\end{abstract}

\section{Introduction}

Co-occurrence network can be modeled efficiently by using hyper-bag-graphs
(hb-graphs for short) introduced in \cite{ouvrard2018adjacency}.
Depending on the information the co-occurrence network carries, the
ranking of the information hold by the associated hb-graph has to
be performed on different features, and the importance stressed on
the lower, higher or medium values. Hence, the necessity of extending
the exchange-based diffusion that is already coupled to a biased random
walk given in \cite{ouvrard2019diffusion} to a more general approach
using biases. We start by giving the background in Section \ref{sec:Mathematical-Background-and}.
We then propose a framework to achieve such a kind of diffusion in
Section \ref{sec:Biased-Diffusion-in} and evaluate it in Section
\ref{sec:Results-and-Evaluation}, before concluding in Section \ref{sec:Further-Comments}.

\section{Mathematical Background and Related Work}

\label{sec:Mathematical-Background-and}

A hb-graph $\mathfrak{H}=\left(V,\mathfrak{E}\right)$ is a family
of multisets $\mathfrak{E}=\left(\mathfrak{e}_{j}\right)_{j\in\left\llbracket p\right\rrbracket }$
of same universe $V=\left\{ v_{i}:i\in\left\llbracket n\right\rrbracket \right\} .$
The elements of $\mathfrak{E}$ are called the hb-edges; each hb-edge
$\mathfrak{e}_{j},$ $j\in\left\llbracket p\right\rrbracket ,$ is
a multiset of universe $V$ and of multiplicity function: $m_{\mathfrak{e}_{j}}:V\rightarrow\mathbb{R}^{+}.$
The m-cardinality $\#_{m}\mathfrak{e}_{j}$ of a hb-edge is: $\#_{m}\mathfrak{e}_{j}=\sum\limits _{i\in\text{\ensuremath{\left\llbracket n\right\rrbracket }}}m_{\mathfrak{e}_{j}}\left(v_{i}\right).$
For more information on hb-graphs, the interested reader can refer
to \cite{ouvrard2019hbgraphs} for a full introduction. A weighted
hb-graph has hb-edges having a weight given by: $w_{e}:\mathfrak{E}\rightarrow\mathbb{R}^{+}.$

In \cite{dehmer2011history}, the authors introduce an abstract information
function $f:V\rightarrow\mathbb{R}^{+}$ which is associated to a
probability for each vertex $v_{i}\in V:$ $p^{f}\left(v_{i}\right)=\dfrac{f\left(v_{i}\right)}{\sum\limits _{j\in\left\llbracket \left|V\right|\right\rrbracket }f\left(v_{j}\right)}.$
In \cite{zlatic2010topologically}, a bias is introduced in the transition
probability of a random walk in order to explore communities in a
network. The bias is either related to a vertex property $x_{i}$
such as the degree or to an edge property $y_{j}$ such as the edge
multiplicity or the shortest path betweenness. For a vertex, the new
transition probability between vertex $v_{i}$ and $v_{j}$ is given
by: $T_{ij}\left(x,\beta\right)=\dfrac{a_{ij}e^{\beta x_{i}}}{\sum\limits _{l}a_{lj}e^{\beta x_{l}}},$where
$A=\left(a_{ij}\right)_{i,j\in\left\llbracket n\right\rrbracket }$
is the adjacency matrix of the graph and $\beta$ is a parameter.

A same kind of bias, can be used related to the edges and can be combined
to the former to have the overall transition probability from one
vertex to another.

\section{Biased Diffusion in Hb-graphs}

\label{sec:Biased-Diffusion-in}

We consider a weighted hb-graph $\mathcal{\mathfrak{H}}=\left(V,\mathfrak{E},w_{e}\right)$
with $V=\left\{ v_{i}:i\in\left\llbracket n\right\rrbracket \right\} $
and $\mathfrak{E}=\left(\mathfrak{e}_{j}\right)_{j\in\left\llbracket p\right\rrbracket };$
we write $H=\left[m_{\mathfrak{e_{j}}}\left(v_{i}\right)\right]_{\substack{i\in\left\llbracket n\right\rrbracket \\
j\in\left\llbracket p\right\rrbracket 
}
}$ the incidence matrix of the hb-graph.

\subsection{Abstract Information Functions and Bias}

\index{abstract information function}We consider a \textbf{hb-edge
based vertex abstract information function: }$f_{V}:V\times E\rightarrow\mathbb{R}^{+}.$
The exchange-based diffusion presented in \cite{ouvrard2018diffusion,ouvrard2019diffusion}
is a particular example of biased diffusion, where the biases are
given in Table \ref{Tab: Features and Bias exchange-based diff}.
An unbiased diffusion would be to have a vertex abstract function
and a hb-edge vertex function that is put to 1 for every vertices
and hb-edges, i.e.\ equi-probability for every vertices and every
hb-edges.

\begin{table}
\begin{center}%
\begin{tabular}{>{\centering}p{7cm}c}
\cellcolor{blue!10}Hb-edge based vertex abstract information function & \cellcolor{cyan!10}$f_{V}\left(v_{i},\mathfrak{e}_{j}\right)=m_{j}\left(v_{i}\right)w\left(\mathfrak{e}_{j}\right)$\tabularnewline
\cellcolor{blue!30}Vertex abstract information function & \cellcolor{cyan!30}$F_{V}\left(v_{i}\right)=d_{w,v_{i}}$\tabularnewline
\cellcolor{blue!10}Vertex bias function & \cellcolor{cyan!10}$g_{V}\left(x\right)=x$\tabularnewline
\cellcolor{blue!30}Vertex overall bias & \cellcolor{cyan!30}$G_{V}\left(v_{i}\right)=d_{w,v_{i}}$\tabularnewline
\cellcolor{blue!10}Vertex-based hb-edge abstract information function & \cellcolor{cyan!10}$f_{E}\left(\mathfrak{e}_{j},v_{i}\right)=m_{j}\left(v_{i}\right)w\left(\mathfrak{e}_{j}\right)$\tabularnewline
\cellcolor{blue!30}Hb-edge abstract information function & \cellcolor{cyan!30}$F_{E}\left(\mathfrak{e}_{j}\right)=w\left(\mathfrak{e}_{j}\right)\#_{m}\mathfrak{e}_{j}$\tabularnewline
\cellcolor{blue!10}Hb-edge bias function & \cellcolor{cyan!10}$g_{E}\left(x\right)=x$\tabularnewline
\cellcolor{blue!30}Hb-edge overall bias & \cellcolor{cyan!30}$G_{E}\left(\mathfrak{e}_{j}\right)=w\left(\mathfrak{e}_{j}\right)\#_{m}\mathfrak{e}_{j}$\tabularnewline
\end{tabular}\end{center}

\caption{Features used in the exchange-based diffusion in \cite{ouvrard2018diffusion}.}

\label{Tab: Features and Bias exchange-based diff}
\end{table}

The \textbf{vertex abstract information function} is defined as the
function: $F_{V}:V\rightarrow\mathbb{R}^{+}$ such that: $F_{V}\left(v_{i}\right)\overset{\Delta}{=}\sum\limits _{j\in\left\llbracket p\right\rrbracket }f_{V}\left(v_{i},\mathfrak{e}_{j}\right).$
The \textbf{probability corresponding to this hb-edge based vertex
abstract information} as: $p^{f_{V}}\left(\mathfrak{e}_{j}|v_{i}\right)\overset{\Delta}{=}\dfrac{f_{V}\left(v_{i},\mathfrak{e}_{j}\right)}{F_{V}\left(v_{i}\right)}.$
If we now consider a \textbf{vertex bias function}: $g_{V}:\mathbb{R}^{+}\rightarrow\mathbb{R}^{+}$
applied to $f_{V}\left(v_{i},\mathfrak{e}_{j}\right),$ we can define
a \textbf{biased probability on the transition from vertices to hb-edges}
as: 
\[
\widetilde{p_{V}}\left(\mathfrak{e}_{j}|v_{i}\right)\overset{\Delta}{=}\dfrac{g_{V}\left(f_{V}\left(v_{i},\mathfrak{e}_{j}\right)\right)}{G_{V}\left(v_{i}\right)},
\]
where $G_{V}\left(v_{i}\right)$, the \textbf{vertex overall bias},
is defined as: $G_{V}\left(v_{i}\right)\overset{\Delta}{=}\sum\limits _{j\in\left\llbracket p\right\rrbracket }g_{V}\left(f_{V}\left(v_{i},\mathfrak{e}_{j}\right)\right).$

Typical choices for $g_{V}$ are: $g_{V}\left(x\right)=x^{\alpha}$
or $g_{V}\left(x\right)=e^{\alpha x}.$ When $\alpha>0$, higher values
of $f_{V}$ are encouraged, and on the contrary, when $\alpha<0$
smaller values of $f_{V}$ are encouraged.

Similarly,the \textbf{vertex-based hb-edge abstract information function}
is defined as the function: $f_{E}:E\times V\rightarrow\mathbb{R}^{+}.$
The \textbf{hb-edge abstract information function} is defined as the
function: $F_{E}:V\rightarrow\mathbb{R}^{+},$ such that: $F_{E}\left(\mathfrak{e}_{j}\right)\overset{\Delta}{=}\sum\limits _{i\in\left\llbracket n\right\rrbracket }f_{E}\left(\mathfrak{e}_{j},v_{i}\right).$
The \textbf{probability corresponding to the vertex-based hb-edge
abstract information} is defined as: $p^{f_{E}}\left(v_{i}|\mathfrak{e}_{j}\right)\overset{\Delta}{=}\dfrac{f_{E}\left(\mathfrak{e}_{j},v_{i}\right)}{F_{E}\left(\mathfrak{e}_{j}\right)}.$
Considering a vertex bias function: $g_{E}:\mathbb{R}^{+}\rightarrow\mathbb{R}^{+}$
applied to $f_{E}\left(\mathfrak{e}_{j},v_{i}\right),$ a \textbf{biased
probability on the transition from hb-edges to vertices} is defined
as:
\[
\widetilde{p_{E}}\left(v_{i}|\mathfrak{e}_{j}\right)\overset{\Delta}{=}\dfrac{g_{E}\left(f_{E}\left(\mathfrak{e}_{j},v_{i}\right)\right)}{G_{E}\left(\mathfrak{e_{j}}\right)},
\]
where the hb-edge overall bias $G_{E}\left(\mathfrak{e}_{j}\right)$
is defined as: $G_{E}\left(\mathfrak{e}_{j}\right)\overset{\Delta}{=}\sum\limits _{i\in\left\llbracket n\right\rrbracket }g_{E}\left(f_{E}\left(\mathfrak{e}_{j},v_{i}\right)\right).$

Typical choices for $g_{E}$ are: $g_{E}\left(x\right)=x^{\alpha}$
or $g_{E}\left(x\right)=e^{\alpha x}.$ When $\alpha>0$, higher values
of $f_{E}$ are encouraged, and on the contrary, when $\alpha<0$
smaller values of $f_{E}$ are encouraged.

\subsection{Biased Diffusion by Exchange}

A two-phase step diffusion by exchange is now considered---with a
similar approach to \cite{ouvrard2018diffusion,ouvrard2019diffusion}---,
taking into account the biased probabilities on vertices and hb-edges.

The vertices hold an information value at time $t$ given by: $\alpha_{t}:\left\{ \begin{array}{c}
V\rightarrow\left[0;1\right]\\
v_{i}\mapsto\alpha_{t}\left(v_{i}\right)
\end{array}\right..$

The hb-edges hold an information value at time $t$ given by: $\epsilon_{t}:\left\{ \begin{array}{c}
\mathfrak{E}\rightarrow\left[0;1\right]\\
\mathfrak{e}_{j}\mapsto\epsilon_{t}\left(\mathfrak{e}_{j}\right)
\end{array}\right..$

We write $P_{V,t}=\left(\alpha_{t}\left(v_{i}\right)\right)_{i\in\left\llbracket n\right\rrbracket }$
the row state vector of the vertices at time $t$ and $P_{\mathfrak{E},t}=\left(\epsilon_{t}\left(\mathfrak{e}_{j}\right)\right)_{j\in\left\llbracket p\right\rrbracket }$
the row state vector of the hb-edges. We call information value of
the vertices, the value: $I_{t}\left(V\right)=\sum\limits _{v_{i}\in V}\alpha_{t}\left(v_{i}\right)$
and $I_{t}\left(\mathfrak{E}\right)=\sum\limits _{\mathfrak{e}_{j}\in\mathfrak{E}}\epsilon_{t}\left(\mathfrak{e}_{j}\right)$
the one of the hb-edges. We write: $I_{t}\left(\mathfrak{H}\right)=I_{t}\left(V\right)+I_{t}\left(\mathfrak{E}\right).$

The initialisation is done such that $I_{0}\left(\mathfrak{H}\right)=1.$
At the diffusion process start, the vertices concentrate uniformly
and exclusively all the information value. Writing $\alpha_{\text{ref}}=\dfrac{1}{\left|V\right|},$
we set for all $v_{i}\in V:$ $\alpha_{0}\left(v_{i}\right)=\alpha_{\text{ref}}$
and for all $\mathfrak{e}_{j}\in\mathfrak{E},$ $\epsilon_{0}\left(\mathfrak{e}_{j}\right)=0.$

At every time step, the first phase starts at time $t$ and ends at
$t+\dfrac{1}{2},$ where values held by the vertices are shared completely
to the hb-edges, followed by the second phase between time $t+\dfrac{1}{2}$
and $t+1$, where the exchanges take place the other way round. The
exchanges between vertices and hb-edges aim at being conservative
on the global value of $\alpha_{t}$ and $\epsilon_{t}$ distributed
over the hb-graph.

\textbf{During the first phase between time $t$ and time $t+\dfrac{1}{2}$},
the contribution to the value $\epsilon_{t+\frac{1}{2}}\left(\mathfrak{e}_{j}\right)$
from the vertex $v_{i}$ is given by:
\[
\delta\epsilon_{t+\frac{1}{2}}\left(\mathfrak{e}_{j}|v_{i}\right)=\widetilde{p_{V}}\left(\mathfrak{e}_{j}|v_{i}\right)\alpha_{t}\left(v_{i}\right)
\]
and:
\[
\epsilon_{t+\frac{1}{2}}\left(\mathfrak{e}_{j}\right)=\sum\limits _{i=1}^{n}\delta\epsilon_{t+\frac{1}{2}}\left(\mathfrak{e}_{j}\mid v_{i}\right).
\]
We have:
\[
\alpha_{t+\frac{1}{2}}\left(v_{i}\right)=\alpha_{t}\left(v_{i}\right)-\sum\limits _{j=1}^{p}\delta\epsilon_{t+\frac{1}{2}}\left(\mathfrak{e}_{j}\mid v_{i}\right).
\]

\begin{claim}[No information on vertices at $t+\dfrac{1}{2}$]

It holds: $\forall i\in\left\llbracket n\right\rrbracket :\alpha_{t+\frac{1}{2}}\left(v_{i}\right)=0.$\end{claim}

\begin{proof}

For all $i\in\left\llbracket n\right\rrbracket :$ 
\[
\alpha_{t+\frac{1}{2}}\left(v_{i}\right)=\alpha_{t}\left(v_{i}\right)\left(1-\dfrac{\sum\limits _{j\in\left\llbracket p\right\rrbracket }g_{V}\left(f_{V}\left(v_{i},\mathfrak{e}_{j}\right)\right)}{G_{V}\left(v_{i}\right)}\right)=0.
\]
\end{proof}

\begin{claim}[Conservation of the information of the hb-graph at $t+\dfrac{1}{2}$]

It holds: $I_{t+\frac{1}{2}}\left(\mathfrak{H}\right)=1.$

\end{claim}

\begin{proof}

We have:
\begin{align*}
I_{t+\frac{1}{2}}\left(\mathfrak{H}\right) & =I_{t+\frac{1}{2}}\left(V\right)+I_{t+\frac{1}{2}}\left(\mathfrak{E}\right)\\
 & =I_{t+\frac{1}{2}}\left(\mathfrak{E}\right)\\
 & =\sum\limits _{\mathfrak{e}_{j}\in\mathfrak{E}}\sum\limits _{i\in\left\llbracket n\right\rrbracket }\delta\epsilon_{t+\frac{1}{2}}\left(\mathfrak{e}_{j}\mid v_{i}\right)\\
 & =\sum\limits _{\mathfrak{e}_{j}\in\mathfrak{E}}\sum\limits _{i\in\left\llbracket n\right\rrbracket }\dfrac{g_{V}\left(f_{V}\left(v_{i},\mathfrak{e}_{j}\right)\right)}{G_{V}\left(v_{i}\right)}\alpha_{t}\left(v_{i}\right)\\
 & =\sum\limits _{i\in\left\llbracket n\right\rrbracket }\alpha_{t}\left(v_{i}\right)\dfrac{\sum\limits _{\mathfrak{e}_{j}\in\mathfrak{E}}g_{V}\left(f_{V}\left(v_{i},\mathfrak{e}_{j}\right)\right)}{G_{V}\left(v_{i}\right)}\\
 & =\sum\limits _{i\in\left\llbracket n\right\rrbracket }\alpha_{t}\left(v_{i}\right)\\
 & =1.
\end{align*}
\end{proof}

We introduce the \textbf{vertex overall bias matrix}: $G_{V}\overset{\Delta}{=}\text{diag}\left(\left(G_{V}\left(v_{i}\right)\right)_{i\in\left\llbracket n\right\rrbracket }\right)$
and the \textbf{biased vertex-feature matrix}: $B_{V}\overset{\Delta}{=}\left[g_{V}\left(f_{V}\left(v_{i},\mathfrak{e}_{j}\right)\right)\right]_{\substack{i\in\left\llbracket n\right\rrbracket \\
j\in\left\llbracket p\right\rrbracket 
}
}.$ It holds:
\begin{equation}
P_{\mathfrak{E},t+\frac{1}{2}}=P_{V,t}G_{V}^{-1}B_{V}.\label{eq:P_E_t_0_5-1}
\end{equation}

\textbf{During the second phase that starts at time $t+\dfrac{1}{2}$},
the values held by the hb-edges are transferred to the vertices. The
contribution to $\alpha_{t+1}\left(v_{i}\right)$ given by a hb-edge
$\mathfrak{e}_{j}$ is proportional to $\epsilon_{t+\frac{1}{2}}$
in a factor corresponding to the biased probability $\widetilde{p_{E}}\left(v_{i}|\mathfrak{e}_{j}\right):$
\[
\delta\alpha_{t+1}\left(v_{i}\mid\mathfrak{e}_{j}\right)=\widetilde{p_{E}}\left(v_{i}|\mathfrak{e}_{j}\right)\epsilon_{t+\frac{1}{2}}\left(\mathfrak{e}_{j}\right).
\]
Hence, we have: $\alpha_{t+1}\left(v_{i}\right)=\sum\limits _{j=1}^{p}\delta\alpha_{t+1}\left(v_{i}\mid\mathfrak{e}_{j}\right)$
and:
\[
\epsilon_{t+1}\left(\mathfrak{e}_{j}\right)=\epsilon_{t+\frac{1}{2}}\left(\mathfrak{e}_{j}\right)-\sum\limits _{i=1}^{n}\delta\alpha_{t+1}\left(v_{i}\mid\mathfrak{e}_{j}\right).
\]

\begin{claim}[The hb-edges have no value at $t+1$]

It holds: $\epsilon_{t+1}\left(\mathfrak{e}_{j}\right)=0.$\end{claim}

\begin{proof}Similar to the one of the first phase for $\alpha_{t+\frac{1}{2}}\left(v_{i}\right).$

\end{proof}

\pagebreak

\begin{claim}[Conservation of the information of the hb-graph at $t+1$]

It holds: $I_{t+1}\left(\mathfrak{H}\right)=1.$\end{claim}

\begin{proof}Similar to the one for the first phase.

\end{proof}

We now introduce $G_{\mathfrak{E}}\overset{\Delta}{=}\text{diag}\left(\left(G_{E}\left(\mathfrak{e}_{j}\right)\right)_{j\in\left\llbracket p\right\rrbracket }\right)$
the diagonal matrix of size $p\times p$ and the \textbf{biased hb-edge-feature
matrix}: $B_{E}\overset{\Delta}{=}\left[g_{E}\left(f_{E}\left(\mathfrak{e}_{j},v_{i}\right)\right)\right]_{\substack{j\in\left\llbracket p\right\rrbracket \\
i\in\left\llbracket n\right\rrbracket 
}
},$ it comes:
\begin{equation}
P_{\mathfrak{E},t+\frac{1}{2}}G_{\mathfrak{E}}^{-1}B_{E}=P_{V,t+1}.\label{eq:P_V_t_1-1}
\end{equation}
Regrouping (\ref{eq:P_E_t_0_5-1}) and (\ref{eq:P_V_t_1-1}):
\begin{equation}
P_{V,t+1}=P_{V,t}G_{V}^{-1}B_{V}G_{\mathfrak{E}}^{-1}B_{E}.\label{eq:P_V_tplus1-1}
\end{equation}
It is valuable to keep a trace of the intermediate state: $P_{\mathfrak{E},t+\frac{1}{2}}=P_{V,t}G_{V}^{-1}B_{V}$
as it records the information on hb-edges.

Writing $T=G_{V}^{-1}B_{V}G_{\mathfrak{E}}^{-1}B_{E}$, it follows
from \ref{eq:P_V_tplus1-1}: $P_{V,t+1}=P_{V,t}T.$

\begin{claim}[Stochastic transition matrix]

$T$ is a square row stochastic matrix of dimension $n.$

\end{claim}

\begin{proof}

Let: $A=\left(a_{ij}\right)_{\substack{i\in\left\llbracket n\right\rrbracket \\
j\in\left\llbracket p\right\rrbracket 
}
}=G_{V}^{-1}B_{V}\in M_{n,p}$ and: $B=\left(b_{jk}\right)_{\substack{j\in\left\llbracket p\right\rrbracket \\
k\in\left\llbracket n\right\rrbracket 
}
}=G_{\mathfrak{E}}^{-1}B_{E}\in M_{p,n}.$ $A$ and $B$ are non-negative rectangular matrices. Moreover:
\begin{itemize}
\item $a_{ij}=\widetilde{p_{V}}\left(\mathfrak{e}_{j}|v_{i}\right)$ and:
$\sum\limits _{j\in\left\llbracket p\right\rrbracket }a_{ij}=\sum\limits _{j\in\left\llbracket p\right\rrbracket }\widetilde{p_{V}}\left(\mathfrak{e}_{j}|v_{i}\right)=\sum\limits _{j\in\left\llbracket p\right\rrbracket }\dfrac{g_{V}\left(f_{V}\left(v_{i},\mathfrak{e}_{j}\right)\right)}{G_{V}\left(v_{i}\right)}=1;$
\item $b_{jk}=\widetilde{p_{E}}\left(v_{k}|\mathfrak{e}_{j}\right)$ and:
$\sum\limits _{k\in\left\llbracket n\right\rrbracket }b_{jk}=\sum\limits _{k\in\left\llbracket n\right\rrbracket }\widetilde{p_{E}}\left(v_{k}|\mathfrak{e}_{j}\right)=\dfrac{\sum\limits _{k\in\left\llbracket n\right\rrbracket }g_{E}\left(f_{E}\left(\mathfrak{e}_{j},v_{k}\right)\right)}{G_{E}\left(\mathfrak{e_{j}}\right)}=1.$
\end{itemize}
We have: $P_{V,t+1}=P_{V,t}AB,$ where: $AB=\left(\sum\limits _{j\in\left\llbracket p\right\rrbracket }a_{ij}b_{jk}\right)_{\substack{i\in\left\llbracket n\right\rrbracket \\
k\in\left\llbracket n\right\rrbracket 
}
}.$

It yields: $\sum\limits _{k\in\left\llbracket n\right\rrbracket }\sum\limits _{j\in\left\llbracket p\right\rrbracket }a_{ij}b_{jk}=\sum\limits _{j\in\left\llbracket p\right\rrbracket }a_{ij}\sum\limits _{k\in\left\llbracket n\right\rrbracket }b_{jk}=\sum\limits _{j\in\left\llbracket p\right\rrbracket }a_{ij}=1.$

Hence $AB$ is a non-negative square matrix with its row sums all
equal to 1: it is a row stochastic matrix.

\end{proof}

\begin{claim}[Properties of T]

Assuming that the hb-graph is connected, the biased feature exchange-based
diffusion matrix $T$ is aperiodic and irreducible.

\end{claim}

\begin{proof}

This stochastic matrix is aperiodic, due to the fact that any vertex
of the hb-graph retrieves a part of the value it has given to the
hb-edge, hence $t_{ii}>0$ for all $i\in\left\llbracket n\right\rrbracket $.
Moreover, as the hb-graph is connected, the matrix is irreducible
as any state can be joined from any other state.

\end{proof}

The fact that $T$ is a stochastic matrix aperiodic and irreducible
for a connected hb-graph ensures that $\left(\alpha_{t}\right)_{t\in\mathbb{N}}$
converges to a stationary state which is the probability vector $\pi_{V}$
associated to the eigenvalue 1 of $T$. Nonetheless, due to the presence
of the different functions for vertices and hb-edges, the simplifications
do not occur anymore as in \cite{ouvrard2018diffusion,ouvrard2019diffusion}
and thus we do not have an explicit expression for the stationary
state vector of the vertices.

The same occurs for the expression of the hb-edge stationary state
vector $\pi_{E}$ which is still calculated from $\pi_{V}$ using
the following formula: $\pi_{E}=\pi_{V}G_{V}^{-1}B_{V}.$

\section{Results and Evaluation}

\label{sec:Results-and-Evaluation}

We consider different biases on a randomly generated hb-graph using
still the same features that in the exchange-based diffusion realized
in \cite{ouvrard2018diffusion,ouvrard2019diffusion}. We generate
hb-graphs with 200 collaborations---built out of 10,000 potential
vertices---with a maximum m-cardinality of 20, such that the hb-graph
has five groups that are generated with two of the vertices chosen
out of a group of 10, that have to occur in each of the collaboration;
there are 20 vertices that have to stand as central vertices, i.e.\
that ensures the connectivity in between the different groups of the
hb-graph.

The approach is similar to the one taken in \cite{ouvrard2018diffusion,ouvrard2019diffusion},
using the same hb-edge based vertex abstract information function
and the same vertex-based hb-edge abstract information function, but
putting different biases as it is presented in Table \ref{Tab: biases used in simulation}.

\begin{table}
\begin{center}%
\begin{tabular}{ccccccc}
\cellcolor{blue!50}Experiment & \cellcolor{cyan!50}1 & \cellcolor{blue!50}2 & \cellcolor{cyan!50}3 & \cellcolor{blue!50}4 & \cellcolor{cyan!50}5 & \tabularnewline
\cellcolor{blue!10}Vertex bias function $g_{V}\left(x\right)=$ & \cellcolor{cyan!10}$x$ & \cellcolor{blue!10}$x^{2}$ & \cellcolor{cyan!10}$x^{0.2}$ & \cellcolor{blue!10}$e^{2x}$ & \cellcolor{cyan!10}$e^{-2x}$ & \tabularnewline
\cellcolor{blue!30}Hb-edge bias function $g_{E}\left(x\right)=$ & \cellcolor{cyan!30}$x$ & \cellcolor{blue!30}$x^{2}$ & \cellcolor{cyan!30}$x^{0.2}$ & \cellcolor{blue!30}$e^{2x}$ & \cellcolor{cyan!30}$e^{-2x}$ & \tabularnewline
 &  &  &  &  &  & \tabularnewline
\cellcolor{blue!50}Experiment & \cellcolor{cyan!50}6 & \cellcolor{blue!50}7 & \cellcolor{cyan!50}8 & \cellcolor{blue!50}9 &  & \tabularnewline
\cellcolor{blue!10}Vertex bias function $g_{V}\left(x\right)=$ & \cellcolor{cyan!10}$x^{2}$ & \cellcolor{blue!10}$e^{2x}$ & \cellcolor{cyan!10}$x^{0.2}$ & \cellcolor{blue!10}$e^{-2x}$ &  & \tabularnewline
\cellcolor{blue!30}Hb-edge bias function $g_{E}\left(x\right)=$ & \cellcolor{cyan!30}$x$ & \cellcolor{blue!30}$x$ & \cellcolor{cyan!30}$x$ & \cellcolor{blue!30}$x$ &  & \tabularnewline
 &  &  &  &  &  & \tabularnewline
\cellcolor{blue!50}Experiment & \cellcolor{cyan!50}10 & \cellcolor{blue!50}11 & \cellcolor{cyan!50}12 & \cellcolor{blue!50}13 & \cellcolor{cyan!50}14 & \cellcolor{blue!50}15\tabularnewline
\cellcolor{blue!10}Vertex bias function $g_{V}\left(x\right)=$ & \cellcolor{cyan!10}$x$ & \cellcolor{blue!10}$x$ & \cellcolor{cyan!10}$x$ & \cellcolor{blue!10}$x$ & \cellcolor{cyan!10}$e^{2x}$ & \cellcolor{blue!10}$e^{-2x}$\tabularnewline
\cellcolor{blue!30}Hb-edge bias function $g_{E}\left(x\right)=$ & \cellcolor{cyan!30}$x^{2}$ & \cellcolor{blue!30}$e^{2x}$ & \cellcolor{cyan!30}$x^{0.2}$ & \cellcolor{blue!30}$e^{-2x}$ & \cellcolor{cyan!30}$e^{-2x}$ & \cellcolor{blue!30}$e^{2x}$\tabularnewline
\end{tabular}\end{center}

\caption{Biases used during the 15 experiments.}

\label{Tab: biases used in simulation}
\end{table}

We compare the rankings obtained on vertices and hb-edges after 200
iterations of the exchange-based diffusion using the strict and large
Kendall tau correlation coefficients for the different biases proposed
in Table \ref{Tab: biases used in simulation}. We present the results
as a visualisation of correlation matrices in Figure \ref{Fig:Biases and rankings vertices-strict kendall}
and in Figure \ref{Fig:Biases and rankings hb-edges-strict Kendall},
lines and columns of these matrices being ordered by the experiment
index presented in Table \ref{Tab: biases used in simulation}. 

We write $\sigma_{i,t}$ the ranking obtained with Experiment $i$
biases for $t\in\left\{ V,E\right\} $ indicating whether the ranking
is performed on vertices or hb-edges---the absence of $t$ means
that it works for both rankings. The ranking obtained by Experiment
1 is called the reference ranking.

In Experiments 2 to 5, the same bias is applied to both vertices and
hb-edges. In Experiments 2 and 3, the biases are increasing functions
on $\left[0;+\infty\right[,$ while in Experiments 4 and 5, they are
decreasing functions. 

Experiments 2 and 3 lead to rankings that are well correlated with
the reference ranking given the large Kendall tau correlation coefficient
value. The higher value of $\tau_{L}\left(\sigma_{i},\sigma_{j}\right)$
compared to the one of $\tau\left(\sigma_{i},\sigma_{j}\right)$ marks
the fact that the rankings with pair of similar biases agree with
the ties in this case. The exponential bias yields to a ranking that
is more granular in the tail for vertices, and reshuffles the way
the hb-edges are ranked; similar observations can be done for both
the vertex and hb-edge rankings in Experiments 2 and 3. 

In Experiments 4 and 5, the rankings remain well correlated with the
reference ranking but the large Kendall tau correlation coefficient
values show that there is much less agreement on the ties, but it
is very punctual in the rankings, with again more discrimination with
an exponential bias. This slight changes imply a reshuffling of the
hb-edge rankings in both cases, significantly emphasized by the exponential
form.

None of these simultaneous pairs of biases reshuffle very differently
the rankings obtained in the head of the rankings of vertices, but
most of them have implications on the head of the rankings of the
hb-edges: typical examples are given in Figure \ref{Fig:Effect_biases_on_rankings_similar}.
It would need further investigations using the Jaccard index.

Dissimilarities in rankings occur when the bias is applied only to
vertices or to hb-edges. The strict Kendall tau correlation coefficients
between the rankings obtained when applying the bias of Experiments
6 to 9---bias on vertices---and 10 to 13---bias on hb-edges---and
the reference ranking for the vertices show weak consistency for vertices
with values around 0.4---Figure \ref{Fig:Biases and rankings vertices-strict kendall}
(a)---, while the large Kendall tau correlation coefficient values
show a small disagreement with values around -0.1---Figure \ref{Fig:Biases and rankings vertices-strict kendall}
(b). For hb-edges, the gap is much less between the strict---values
around 0.7 as shown in Figure \ref{Fig:Biases and rankings hb-edges-strict Kendall}
(a)---and large Kendall tau correlation coefficient values---with
values around 0.6 as shown in Figure \ref{Fig:Biases and rankings hb-edges-strict Kendall}
(b).

Biases with same monotony variations---$g_{t}\left(x\right)=x^{2}$
and $g_{t}\left(x\right)=e^{2x}$ on the one hand and $g_{t}\left(x\right)=x^{0.2}$
and $g_{t}\left(x\right)=e^{-2x}$ on the other hand---have similar
effects independently of their application to vertices xor to hb-edges.
It is also worth to remark that increasing biases lead to rankings
that have no specific agreement or disagreement with rankings of decreasing
biases---as it is shown with $\tau_{L}\left(\sigma_{i},\sigma_{j}\right)$
and $\tau\left(\sigma_{i},\sigma_{j}\right)$ for $i\in\left\llbracket 6;13\right\rrbracket .$

We remark also that increasing biases applied only to vertices correlate
with the corresponding decreasing biases applied only to hb-edges,
and vice-versa. This is the case for Experiments 6 and 12, Experiments
7 and 13, Experiments 8 and 10, and Experiments 9 and 11 for both
vertices---Figures \ref{Fig:Biases and rankings vertices-strict kendall}
(a) and (b)---and hb-edges---Figures \ref{Fig:Biases and rankings hb-edges-strict Kendall}
(a) and (b).

Finally, we conduct two more experiments---Experiments 14 and 15---combining
the biases $g_{t}\left(x\right)=e^{2x}$ and $g_{t}\left(x\right)=e^{-2x}$
in two different manners. With no surprise, they reinforce the disagreement
with the reference ranking both on vertices and hb-edges, with a stronger
disagreement when the decreasing bias is put on vertices. We can remark
that Experiment 14---$g_{V}\left(x\right)=e^{2x}$ and $g_{E}\left(x\right)=e^{-2x}$---has
the strongest correlations with the rankings of dissimilar biases
that are either similar to the one of vertices---Experiments 6 and
7--- or to the one of hb-edges---Experiments 12 and 13.

\begin{figure}
\begin{center}\resizebox{\textwidth}{!}{%
\begin{tabular}{>{\centering}p{1\textwidth}}
(a) Strict Kendall tau correlation coefficient\tabularnewline
\includegraphics[height=0.45\textheight]{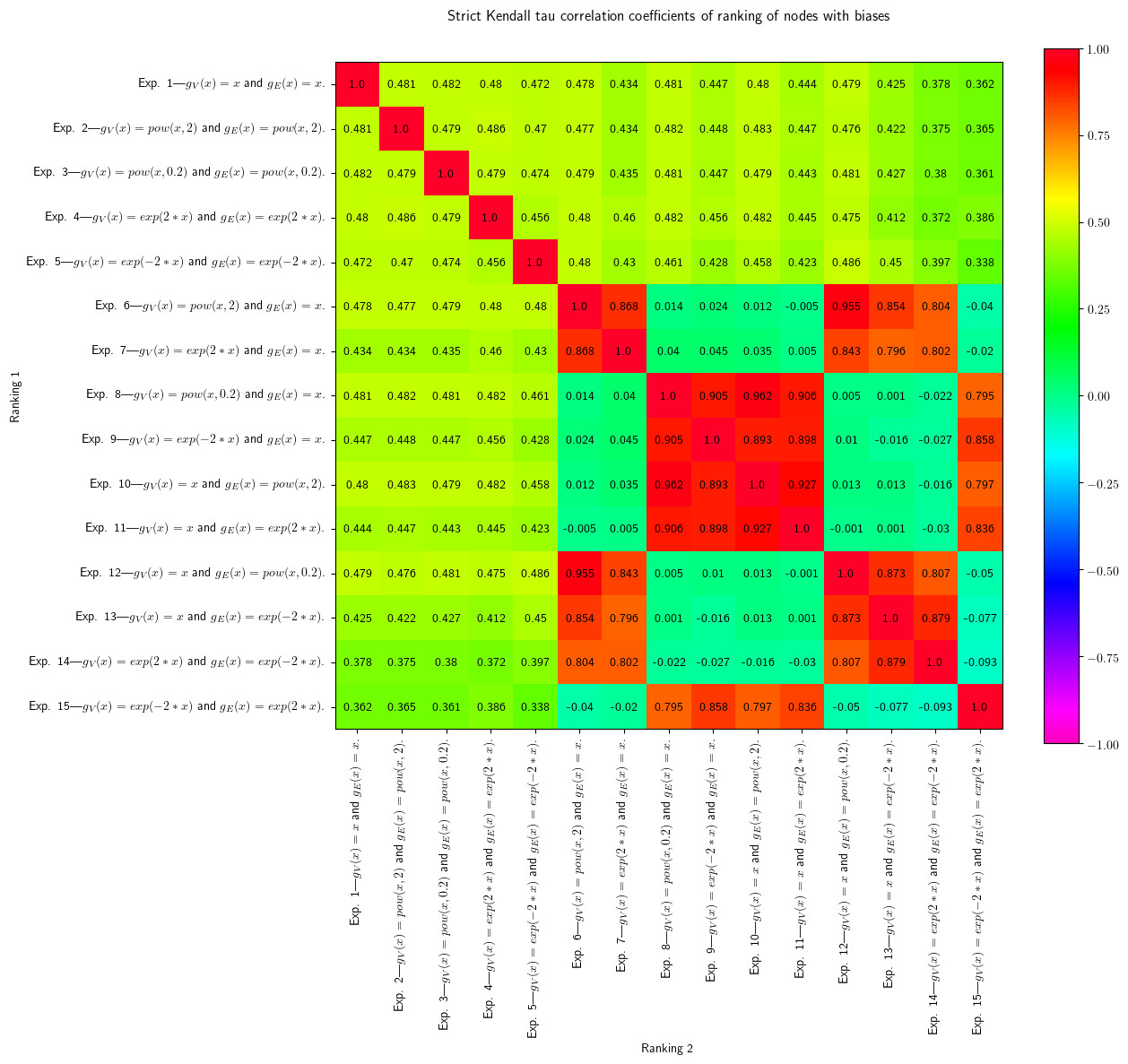}\tabularnewline
(b) Large Kendall tau correlation coefficient\tabularnewline
\includegraphics[height=0.45\textheight]{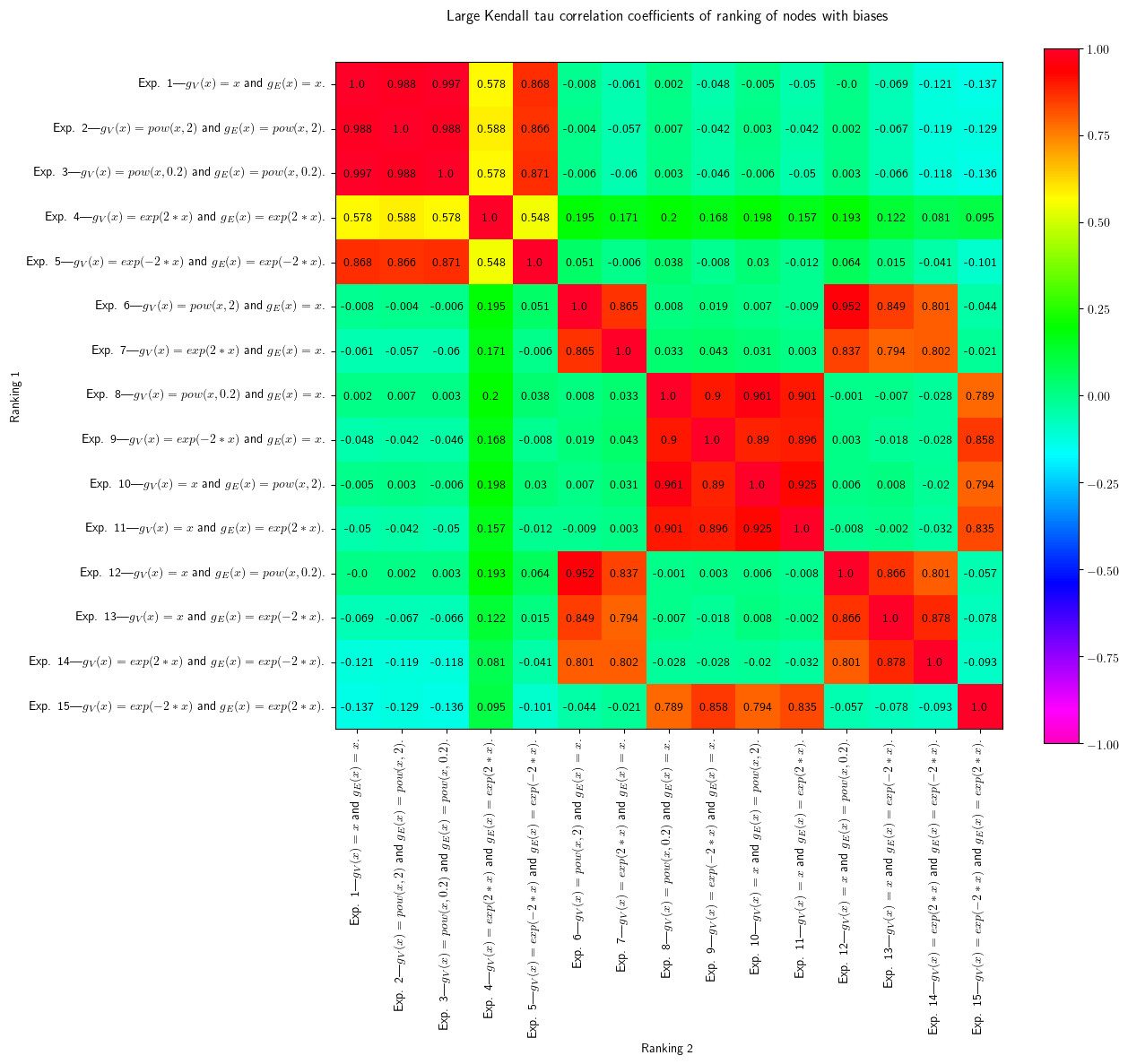}\tabularnewline
\end{tabular}}\end{center}

\caption{Strict (a) and large (b) Kendall tau correlation coefficient for node
ranking with biases. Realized on 100 random hb-graphs with 200 hb-edges
of maximal size 20, with 5 groups.}

\label{Fig:Biases and rankings vertices-strict kendall}
\end{figure}

\begin{figure}
\begin{center}\resizebox{\textwidth}{!}{%
\begin{tabular}{>{\centering}p{1\textwidth}}
(a) Strict Kendall tau correlation coefficient\tabularnewline
\includegraphics[height=0.45\textheight]{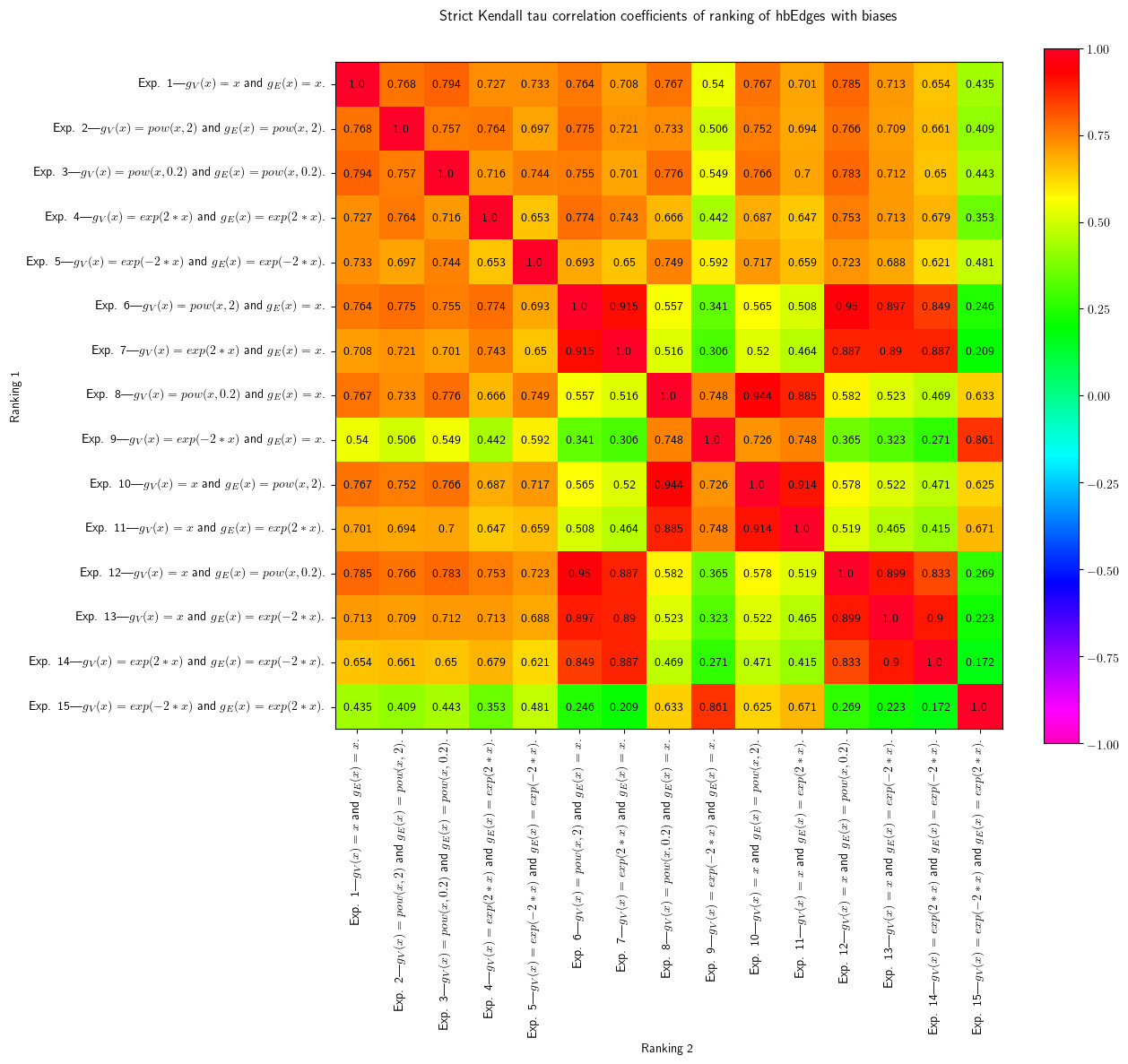}\tabularnewline
(b) Large Kendall tau correlation coefficient\tabularnewline
\includegraphics[height=0.45\textheight]{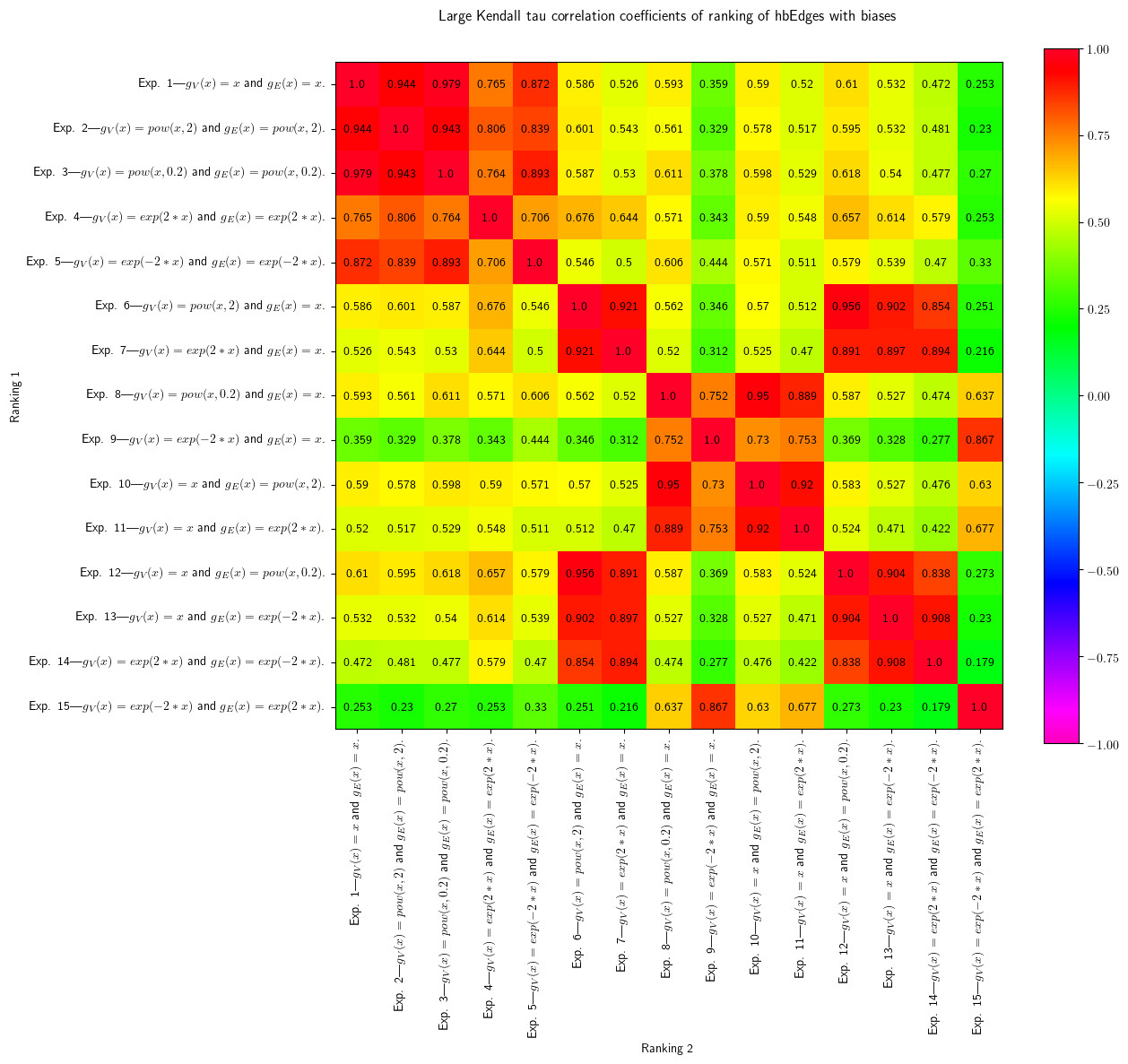}\tabularnewline
\end{tabular}}\end{center}

\caption{Strict (a) and large (b) Kendall tau correlation coefficient for hb-edge
ranking with biases. Realized on 100 random hb-graphs with 200 hb-edges
of maximal size 20, with 5 groups.}

\label{Fig:Biases and rankings hb-edges-strict Kendall}
\end{figure}

\begin{figure}
\begin{center}\resizebox*{!}{\dimexpr\textheight-2\baselineskip\relax}{%
\begin{tabular}{>{\centering}p{1\textwidth}}
\includegraphics[width=1\textwidth]{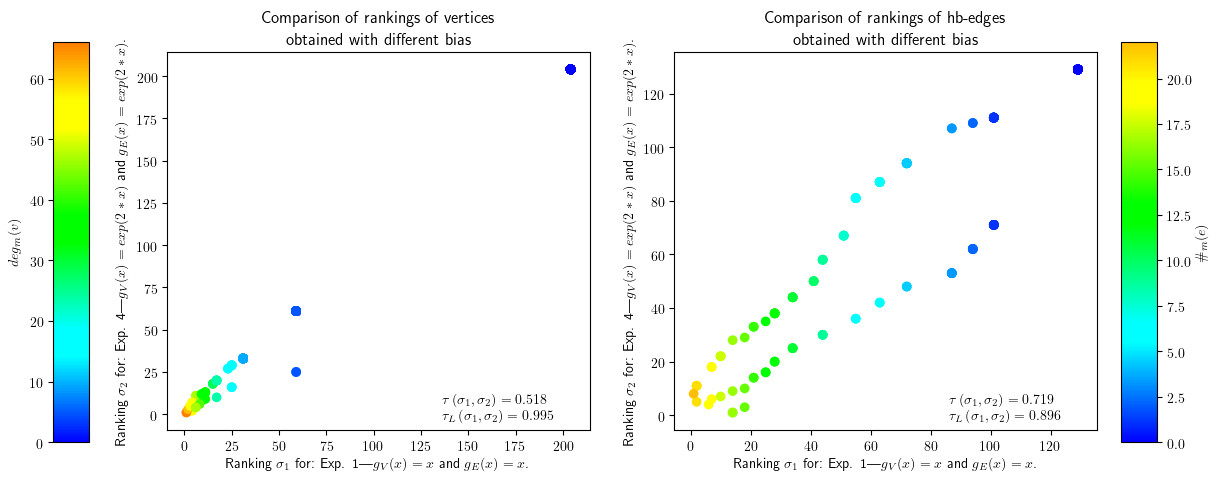}\tabularnewline
{\footnotesize{}(a) First ranking: $g_{V}\left(x\right)=x$ and $g_{E}\left(x\right)=x$;
Second ranking: $g_{V}\left(x\right)=e^{2x}$ and $g_{E}\left(x\right)=e^{2x}.$}\tabularnewline
\includegraphics[width=1\textwidth]{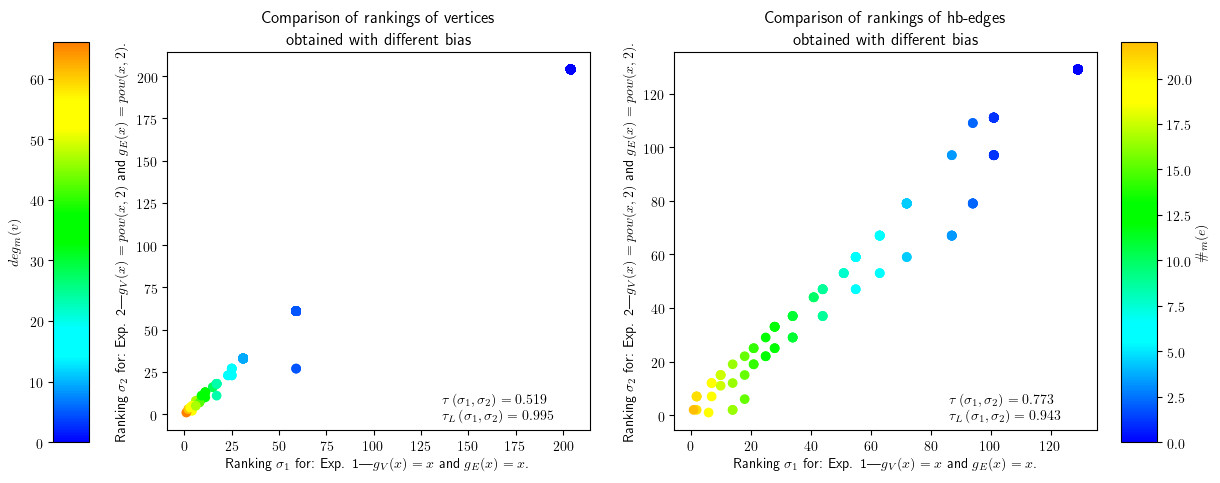}\tabularnewline
{\footnotesize{}(b) First ranking: $g_{V}\left(x\right)=x$ and $g_{E}\left(x\right)=x$;
Second ranking: $g_{V}\left(x\right)=x^{2}$ and $g_{E}\left(x\right)=x^{2}.$}\tabularnewline
\includegraphics[width=1\textwidth]{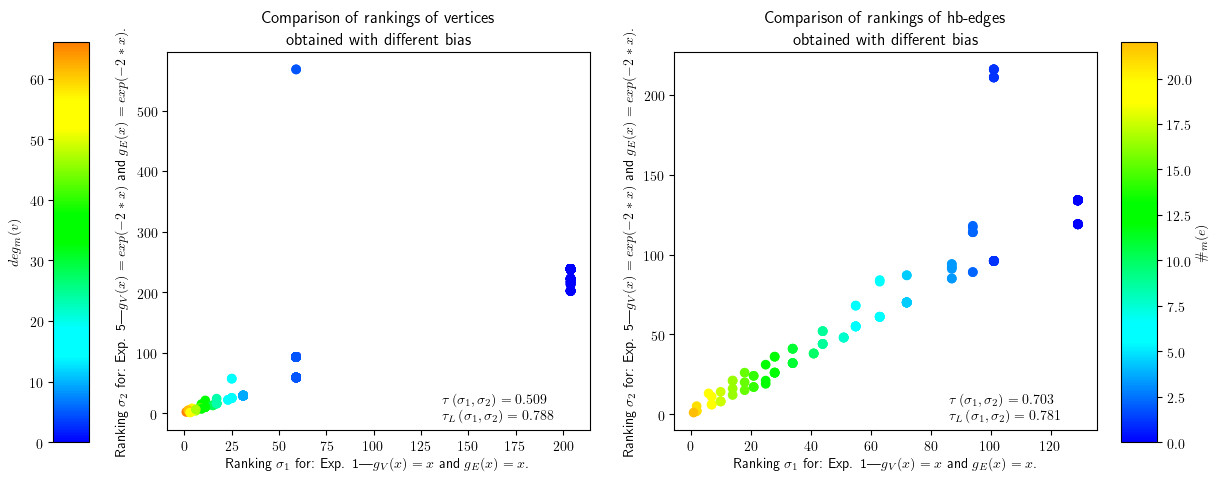}\tabularnewline
{\footnotesize{}(c) First ranking: $g_{V}\left(x\right)=x$ and $g_{E}\left(x\right)=x$;
Second ranking: $g_{V}\left(x\right)=e^{-2x}$ and $g_{E}\left(x\right)=e^{-2x}.$}\tabularnewline
\includegraphics[width=1\textwidth]{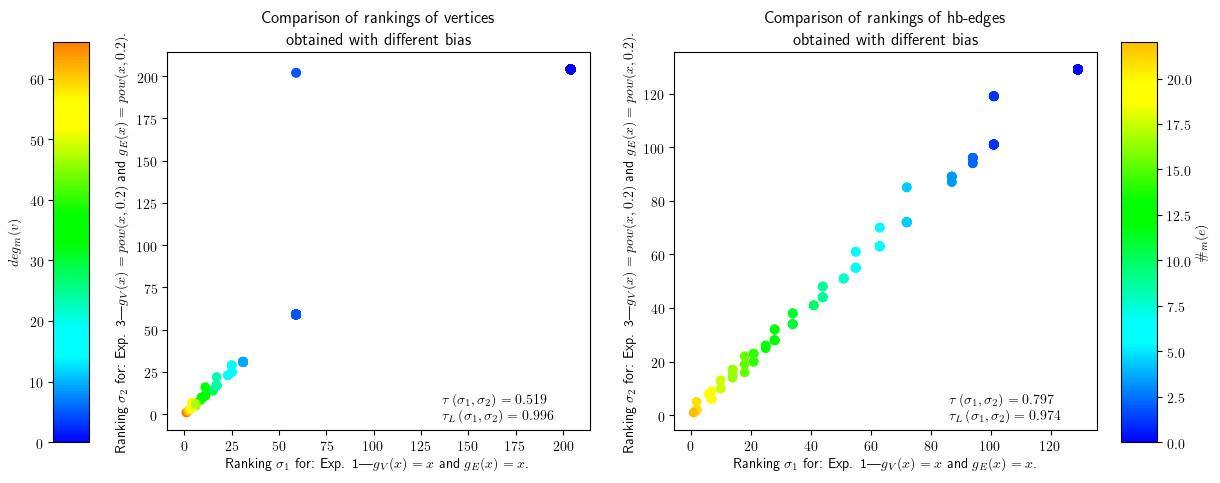}\tabularnewline
{\footnotesize{}(d) First ranking: $g_{V}\left(x\right)=x$ and $g_{E}\left(x\right)=x$;
Second ranking: $g_{V}\left(x\right)=x^{0.2}$ and $g_{E}\left(x\right)=x^{0.2}.$}\tabularnewline
\end{tabular}}\end{center}

\caption{Effect of vertex biases on ranking.}

\label{Fig:Effect_biases_on_rankings_similar}
\end{figure}

A last remark is on the variability of the results: if the values
of the correlation coefficients change, from one hb-graph to another,
the phenomenon observed remains the same, whatever the first hb-graph
observed; however, the number of experiments performed ensures already
a minimized fluctuation in these results.

\section{Further Comments}

\label{sec:Further-Comments}

The biased-exchange-based diffusion proposed in this Chapter enhances
a tunable diffusion that can be integrated into the hb-graph framework
to tune adequately the  ranking of the facets. The results obtained
on randomly generated hb-graphs have still to be applied to real hb-graphs,
with the known difficulty of the connectedness: it will be addressed
in future work. There remains a lot to explore on the subject in order
to refine the query results obtained with real searches. The difficulty
remains that in ground truth classification by experts, only a few
criteria can be retained, that ends up in most cases in pairwise comparison
of elements, and, hence, does not account for higher order relationships.

\bibliographystyle{ieeetr}

\end{document}